\documentclass{article}

\usepackage{appendix}
\usepackage{amsfonts,amsmath,amssymb,dsfont}
\usepackage{amsthm}
\usepackage{cite,url,balance}
\usepackage{algorithm,algorithmic}
\usepackage{etoolbox}
\usepackage{graphicx,psfrag}
\usepackage[usenames,dvipsnames]{color}
\usepackage{array,multirow}

\newtheorem{theorem}{Theorem}
\newtheorem{definition}{Definition}

\newtheorem{cor}{Corollary}
\newtheorem{corollary}[theorem]{Corollary}
\newtheorem{lemma}{Lemma}

\newfont{\mbb}{msbm10 scaled 1100}

\definecolor{red}{RGB}{255,0,0}

\definecolor{dred}{RGB}{210,20,50}

\definecolor{green}{RGB}{0,255,0}

\definecolor{dgreen}{RGB}{0,151,0}

\definecolor{lgreen}{RGB}{75,180,130}

\definecolor{blue}{RGB}{0,0,255}

\definecolor{lblue}{RGB}{30,160,240}

\definecolor{magenta}{RGB}{255,0,255}

\definecolor{orange}{RGB}{255,128,0}

\definecolor{dorange}{RGB}{255,165,0}

\definecolor{violet}{RGB}{207,74,221}

\definecolor{grey}{RGB}{191,191,191}

\def\MSE{{\rm MSE}} 
\def\diag{{\rm diag}} 
\def\tr{{\rm tr}} 
\def\E{{\rm E}} 
\def\Pr{{\rm Pr}} 

\def\N{{\mathds{N}}} 
\def\Z{{\mathds{Z}}} 
\def\R{{\mathds{R}}} 
\def\C{{\mathds{C}}} 

\newcommand{\alg}{\mathcal{A}}
\newcommand{\eps}{\varepsilon}
\renewcommand{\vec}[1]{\mathbf{#1}}

\newcommand{\junk}[1]{}

\DeclareMathOperator{\Lap}{Lap}

\def\b0{{\bf 0}}

\def\ba{{\bf a}}
\def\bb{{\bf b}}

\def\bof{{\bf f}}

\def\bh{{\bf h}}

\def\bx{{\bf x}}
\def\by{{\bf y}}
\def\bz{{\bf z}}

\def\bA{{\bf A}}

\def\bF{{\bf F}}

\def\bH{{\bf H}}
\def\bI{{\bf I}}

\DeclareMathOperator{\herdisc}{herdisc}
\DeclareMathOperator{\specLB}{specLB}
\DeclareMathOperator{\polylog}{polylog}

\begin{document}

\title{Nearly Optimal Private Convolution}
\author{Nadia Fawaz\thanks{Technicolor, Palo Alto CA,
    nadia.fawaz@technicolor.com} \and S. Muthukrishnan\thanks{Rutgers
    University, Piscataway NJ, muthu@cs.rutgers.edu} \and Aleksandar
  Nikolov\thanks{Rutgers University, Piscataway NJ, anikolov@cs.rutgers.edu}}

\maketitle
\begin{abstract}

We study computing the convolution of a private input $x$ with a public input $h$, while satisfying the guarantees of $(\eps, \delta)$-differential privacy. Convolution is a fundamental operation, intimately related to Fourier Transforms.
,
In our setting, the private input may represent a time series of sensitive events or a histogram of a database of confidential personal information. Convolution then captures  important primitives including linear filtering, which is an essential tool in time series analysis, and aggregation queries on projections of the data.

We give a nearly optimal algorithm for computing convolutions while satisfying $(\eps, \delta)$-differential privacy. 
Surprisingly, we follow the simple strategy of adding independent Laplacian noise to each Fourier coefficient and bounding the privacy loss using the composition theorem from~\cite{dwork2010boosting}. We derive a closed form expression for the optimal noise to add to each Fourier coefficient using convex programming duality. Our algorithm is very efficient -- it is essentially no more computationally expensive than a Fast Fourier Transform. 
To prove near optimality, we use the recent discrepancy lowerbounds of~\cite{stoc} and derive a spectral lower bound using  a characterization of discrepancy in terms of determinants. 
\end{abstract}

\section{Introduction}

The \emph{noise complexity} of linear queries is of fundamental interest in
the theory of differential privacy. Consider a database that
represents users (or events) of $N$ different types (in the case of
events, a type is a time step). We may encode the
database as a vector $\vec{x}$ indexed by $\{1, \ldots, N\}$, where
$x_i$ gives the number of users of type $i$. A linear query asks for
the dot product $\langle \vec{a}, \vec{x} \rangle$; a \emph{workload}
of $M$ queries is given as a matrix $\vec{A}$, and the intended output
is $\vec{Ax}$. As the database often encodes personal information, we
wish to answer queries in a way that does not compromise the
individuals represented in the data. We adopt the now standard notion
of \emph{$(\eps, \delta)$-differential privacy}~\cite{DMNS}; informally, an
algorithm is differentially private if its output distribution does
not change drastically when a single user/event changes in the
database. This definition necessitates randomizition and
approximation, and, therefore, the question of the optimal accuracy of
any differentially private algorithm on a workload $\vec{A}$ comes
into the center. We discuss accuracy in terms of \emph{mean squared error} as
a measure of approximation: the expected average of squared error over
all $M$ queries.

The queries in a workload $\vec{A}$ can have different degrees of
correlation, and this poses different challenges for the private
approximation algorithm. In one extreme, when $\vec{A}$ is a set of
$\Omega(N)$ independently sampled random $\{0, 1\}$ (i.e.~counting)
queries, we know, by the seminal work of Dinur and
Nissim~\cite{dinur2003revealing}, that any $(\eps,
\delta)$-differentially private algorithm needs to incur at least
$\Omega(N)$ squared error per query on average. On the other hand, if
$\vec{A}$ consists of the same counting query repeated $M$ times, we
only need to add $O(1)$ noise per query~\cite{DMNS}. While those two
extremes are well understood -- the bounds cited above are tight --
little is known about workloads of queries with some, but not perfect,
correlation.

The \emph{convolution}\footnote{Here we define circular convolution,
  but, however, as discussed in the paper, our results generalize to other
  types of convolution, which are defined similarly.} \junk{(and convolution in general)} of
the private input $\vec{x}$
with a public vector $\vec{h}$ is defined as the vector $\vec{y}$
where 
$$y_i = \sum_{j = 1}^N{h_jx_{i - j \pmod N}}.$$
This convolution map is a workload of $N$ linear queries. Each query
is a circular shift of the previous one, and, therefore, the queries
are far from independent but not identical either. Convolution is a
fundamental operation that arises in algebraic computations such as 
polynomial multiplication. 
It is  a basic operation in signal
analysis and has well known connection to Fourier transforms.  Of primary interest to us, it is a
natural primitive in various applications:
\begin{itemize}
\item linear filters in the analysis of time series data can be cast
  as convolutions; as example applications, linear filtering
  can be used to isolate cycle components in time series data from
  spurious variations, and to compute time-decayed statistics of the
  data; 
\item when user type in the database is specified by $d$ binary
  attributes, aggregate queries such as $k$-wise marginals and
  generalizations can be represented as convolutions. 
\end{itemize}
Privacy concerns arise naturally in these applications: the time
series data can contain records of sensitive events, such as financial
transactions, records of user activity, etc.; some of the attributes
in a database can be sensitive, for example when dealing with
databases of medical data. \junk{Thus in studying differential privacy of
linear queries, the set corresponding to convolutions is a
particularly important case, from foundational and application points
of view.}

We give the first nearly optimal algorithm for computing convolution
under $(\eps, \delta)$-differential privacy constraints.  Our
algorithm gives the lowest mean squared error achievable by adding
independent (but non-uniform) Laplace noise to the Fourier
coefficients of $\vec{x}$ and bounding the privacy loss by the
composition theorem of Dwork et al.~\cite{dwork2010boosting}. Using
complementary slackness conditions, we derive a simple closed form for
the optimal amount of error that should be added in the direction of
each Fourier coefficient. We prove that, for any fixed $\vec{h}$, up
to polylogarithmic factors, \emph{any $(\eps, \delta)$-differential
  private algorithm} incurs at least as much squared error per query
as our algorithm. \emph{Somewhat surprisingly, our result shows that
  the simple strategy of adding indepdendent noise in the Fourier
  domain is nearly optimal for computing convolutions}. Prior to our
work there were known nearly instance-optimal\footnote{Note that
  instance-optimality here refers to the query vector $\vec{h}$, while
  we still consider worst-case error over the private input
  $\vec{x}$.} ($\eps, \delta$)-differentially private algorithm for a
natural class of linear queries. Additionally, our algorithm is
simpler and more efficient than related algorithms for $(\eps,
0)$-differential privacy. 

\junk{
Subsequently to our work, our results were generalized
by Nikolov, Talwar, and Zhang~\cite{geometry} to give instance-optimal
$(\eps, \delta)$-differentially a private algorithm for any workload
of linear queries. However, the generality of their result comes at a
cost in efficiency: their algorithm needs to compute a minimum
enclosing ellipsoid of a convex body and has running time at least
$O(M^2N)$. 
}

To prove optimality of our algorithm, we use the recent
discrepancy-based noise lower bounds of Muthukrishnan and
Nikolov~\cite{stoc}.  We use a
characterization of discrepancy in terms of determinants of
submatrices discovered by Lov\'{a}sz, Spencer, and Vesztergombi, together
with ideas by Hardt and Talwar, who give instance-optimal algorithms
for the stronger notion of $(\eps, 0)$-differential
privacy\footnote{Note that establishing instance-optimality for
  $(\eps, \delta)$-differential privacy is harder from error lower
  bounds perspective, as the privacy definition is weaker.}. A main
technical ingredient in our proof is a connection between the
discrepancy of a matrix $\vec{A}$ and the discrepancy of $\vec{PA}$
where $\vec{P}$ is an orthogonal projection operator.

In addition to applications to linear filtering, our algorithm allows
us to approximate marginal queries encoded by $w$-DNFs, which
generalize $k$-wise marginal queries. Using concentration results for
the spectrum of bounded-width DNFs, we derive a non-trivial error
bound for approximating $w$-DNF queries. The bound  is independent of the DNF
size.

\textbf{Related work}. The problem of computing private
convolutions has not been considered in the literature
before. However, there is a fair amount of work on the more general
problem of computing arbitrary linear queries, as well as some work on
special cases of convolution maps. 

The problem of computing arbitrary linear maps of a private database
histogram was first considered in the seminal work of Dinur and
Nissim~\cite{dinur2003revealing}. They showed that privately answering
$M$ random 0-1 queries on a universe of size $N$ requires $\Omega(N)$
mean squared error as long as $M = \Omega(N)$, and this bound is
tight. These bounds do not directly apply to our work, as a set of
independent random queries is not likely to encode a circular
convolution. Nevertheless, one can show, using spectral noise lower
bounds, that a convolution with a random 0-1 vector $h$ requires
assymptotically as much error as $N$ random queries. Yet, many
particular convolutions of interest require much less noise. This fact
motivates us to study algorithms for approximating the convolution $x
\ast h$ which are optimal 
for any given ${h}$. An efficient algorithm with this kind of instance
per instance (in terms of $h$) optimality gaurantee obviates the need
to develop specialized algorithms. Next we review some prior work on
special instances of convolution maps and also related work on
computing linear maps  optimally.

Bolot et al.~\cite{bolot2011private} give algorithms for various
decayed sum queries: window sums, exponentially and polynomially
decayed sums. Any decayed sum function is a type of linear filter,
and, therefore, a special case of convolution. Thus, our current work
gives a nearly optimal $(\eps, \delta)$-differentially private
approximation for \emph{any decayed sum function}. Moreover, as far as
mean squared error is concerned, our algorithms give improved error
bounds for the window sums problem: constant squared error per query.
However, unlike~\cite{bolot2011private}, we only consider the offline
batch-processing setting, as opposed to the online continual
observation setting.

The work of Barak et al.~\cite{barak2007privacy} on computing $k$-wise
marginals concerns a restricted class of convolutions (see
Section~\ref{sect:gen-apps}). Moreover,
Kasiviswanathan~\cite{kasiviswanathan2010price} show a noise lower
bound for $k$-wise marginals which is tight in the worst case. Our
work is a generalization: we are able to give nearly optimal
approximations to a wider class of queries, and our lower and upper
bounds nearly match for any convolution.

Li and Miklau~\cite{LiM12a,LiM12b} proposed the class of extended
matrix mechanisms, building on prior work on the matrix
mechanism~\cite{matrixmech}, and showed how to efficently compute the
optimal mechanism from the class. Furthermore, independently and
concurrently with our work, Cormode et al.~\cite{cormode2012}
considered adding optimal non-uniform noise to a fixed transform of
the private database. Since our mechanism is a special instance of the
extended matrix mechanism, the algorithms of Li and Miklau have
at most as much error as our algorithm. However, similarly
to~\cite{cormode2012}, we gain significantly in efficiency by fixing a
specific transform (in our case the Fourier transform) of the data and
computing a closed form expression for the optimal noise
magnitudes. Unlike the work of Li and Miklau and Cormode et al., we
are able to show nearly tight \emph{lower bounds} for \emph{any}
differentially private algorithm (not just the extended matrix
mechanism) and any set of convolution queries. Therefore, we can show
that the choice of the Fourier transform comes without loss of
generality for any set of convolution queries.

In the setting of $(\epsilon, 0)$-differential privacy, Hardt and
Talwar~\cite{hardt2010geometry} prove nearly optimal upper and lower
bounds on approximating $\vec{Ax}$ for any matrix $\vec{A}$. Recently,
their results were improved, and made unconditional by Bhaskara et
al.~\cite{12vollb}\junk{ (the original upper bound result relied on a
  conjecture in convex geometry)}.  Prior to our work a similar result
was not known for the weaker notion of approximate privacy,
i.e. $(\eps, \delta)$-differential privacy.  Subsequently to our work,
our results were generalized by Nikolov, Talwar, and
Zhang~\cite{geometry} to give nearly optimal algorithms for computing
any linear map $A$ under $(\eps,\delta)$-differential privacy. Their
work combined our use of hereditary discrepancy bounds on error
through the determinant lower bound with results from assymptotic
convex geometry.  The algorithms from~\cite{hardt2010geometry,12vollb}
are computationally expensive, as they need to sample from a
high-dimensional convex body\footnote{One of the best known algorithms
  is due to Lov\'{a}sz and Vempala~\cite{lovasz2006fast} and, ignoring
  other parameters, makes $\Theta(N^3)$ calls to a separation oracle,
  each of which would require solving a linear programming feasibility
  problem.}. Even the more efficient algorithm from~\cite{geometry}
has running time $\Omega(N^3)$, as it needs to approximate the minimum
enclosing ellipsoid of an $N$-dimensional convex body. By contrast our
algorithm's running time is dominated by the running time of the Fast
Fourier Transform, i.e. $O(N\log N)$, making it more suitable for
practical applications. Also, for some sets of queries, such as
running sums, our analysis gives tighter bounds than the analysis of
the algorithm in~\cite{geometry}.

A related line of work seeks to exploit sparsity assumptions on the
private database in order to reduce error; as we do not limit the
database size, our results are not directly comparable. Using our
histogram representation, database size corresponds to the norm
$\|\vec{x}\|_1$ where $\vec{x}$ is the database in histogram
representation. For general linear queries, the multiplicative weights
algorithm of Hardt and Rothblum achieves mean squared error $O(n
\sqrt{\log N})$ for $\|x\|_1 \leq n$. This bound is nearly tight for
random queries, but can be loose for special queries of interest. For
example, running sums require noise $O(\log^{O(1)} N)$, which is less
than $n$ except for $n$ very small in the universe size. In general,
algorithms which bound database size in order to bound error become
less useful when database size is large compared to the total number
of queries, and for very large databases algorithms such as ours are
still of interest. This is true also for the line of algorithms for
marginal queries which give error an arbitrary small constant fraction
of the database
size~\cite{gupta2011privately,hardt2012private,cheraghchi2012submodular,thaler2012faster}. Note
further that the optimal error for~\emph{a subset of all marginal
  queries} may be less than linear in database size, and our
algorithms will give near optimal error for the specific subset of
interest.

\textbf{Organization}. We begin with preliminaries on differential
privacy and convolution operators. In section~\ref{sec:LB} we derive
our main lower bound result, and in section~\ref{sec:UB} we describe
and analyze our nearly optimal algorithm. In
section~\ref{sect:gen-apps} we describe applications of our main
results.

\section{Preliminaries}\label{sect:prelims}


\emph{Notation:} $\N$, $\R$, and $\C$ are the sets of non-negative
integers, real, and complex numbers respectively. By $\log$ we denote
the logarithm in base $2$ while by $\ln$ we denote the logarithm in
base $e$.  Matrices and vectors are represented by boldface upper and
lower cases, respectively. $\bA^T$, $\bA^\ast$, $\bA^H$ stand for the
transpose, the conjugate and the transpose conjugate of $\bA$,
respectively. The trace and the determinant of $\bA$ are respectively
denoted by $\tr(\bA)$ and $\det(\bA)$. $\bA_{m:}$ denotes the $m$-th
row of matrix $\bA$, and $\bA_{:n}$ its $n$-th column. $\bA|_S$, where
$\bA$ is a matrix with $N$ columns and $S \subseteq [N]$, denotes the
submatrix of $\bA$ consisting of those columns corresponding to
elements of $S$.  $\lambda_{\bA}(1), \ldots, \lambda_{\bA}(n)$
represent the eigenvalues of an $n \times n$ matrix $\bA$. $\bI_N$ is
the identity matrix of size $N$. $\E[\cdot]$ is the statistical
expectation operator. $Lap(x, s)$ denotes the Laplace distribution
centered at $x$ with scale $s$, i.e. the distribution of the random
variable $x + \eta$ where $\eta$ has probability density function
 $p(y) \propto \exp(-|y|/s)$.

\subsection{Convolution}

In this section, we first give the definition of circular
convolution. We then recall important results on the Fourier
eigen-decomposition of convolution. Generalization to other notions of
convolution and applications are discussed in Section~\ref{sect:gen-apps}.


Let $x=\{x_0,\ldots,x_{N-1}\}$ be a real input sequence of length $N$, and $h=\{h_0,\ldots,h_{N-1}\}$ a
sequence of length $N$.
%
The circular convolution of $x$ and $h$ is the sequence $y= x \ast h$ of length $N$ defined by
\begin{equation}\label{eq:circConv}
y_k =  \sum_{n=0}^{N-1} x_n h_{(k-n)\bmod{N}}\mbox{, } \forall k\in\{0,\ldots, N-1\}.
\end{equation}

\begin{definition}
The $N \times N$ circular convolution matrix $\bH$ is defined as
\begin{equation}
\bH=\left[
            \begin{array}{*{6}{c}}
            h_0       & h_{N-1} & h_{N-2} & \ldots  & h_1      \\
            h_1       & h_0     & \ddots  & \ddots  & \vdots   \\
            h_2       & \ddots  & \ddots  & \ddots  & h_{N-2}  \\
            \vdots    & \ddots  & \ddots  & h_0     & h_{N-1}  \\
            h_{N-1}   & \ldots  & h_2     & h_1     & h_0
            \end{array}
          \right]_{N \times N}. \nonumber
\end{equation}
This matrix is a circulant matrix with first column $\bh=[h_0,\ldots,h_{N-1}]^T \in \R^N$, and its subsequent columns are successive cyclic shifts of its first column.
Note that $\bH$ is a normal matrix ($\bH\bH^H = \bH^H \bH$).
\end{definition}

Define the column vectors $\bx=[x_0,\ldots,x_{N-1}]^T \in \R^N$, and $\by=[y_0,\ldots,y_{N-1}]^T \in \R^N$.  The circular convolution (\ref{eq:circConv}) can be written in matrix notation
$\by = \bH \bx$.
In Section~\ref{sec:FourierConv}, we recall that circular convolution can be diagonalized in the Fourier basis.

\subsection{Fourier Eigen-decomposition of Convolution}\label{sec:FourierConv}

In this section, we recall the definition of the Fourier basis, and the eigen-decomposition of circular convolution in this basis.

\begin{definition}
The normalized Discrete Fourier Transform (DFT) matrix of size $N$ is defined as
\begin{equation}\label{eq:NormDFT}
\bF_N=\left\{\frac{1}{\sqrt{N}}\exp{ \left(- \frac{j 2 \pi \: m \: n}{N}\right)}\right\}_{m,n \in \{0,\ldots,N-1\}}.
\end{equation}
Note that $\bF_N$ is symmetric ($\bF_N = \bF_N^T$) and unitary ($\bF_N \bF_N^H=\bF_N^H \bF_N =\bI_N$).
\end{definition}
We denote by $\bof_m =[1,e^{\frac{j 2 \pi \: m }{N}},\ldots,e^{ \frac{j 2 \pi \: m \:  (N-1)}{N}}]^T \in \C^N$ the $m$-th column of the inverse DFT matrix $\bF_N^H$. Or alternatively, $\bof_m^H$ is the $m$-th row of $\bF_N$.
The normalized DFT of a vector $\bh$ is simply given by
$\hat{\bh}=\bF_N \bh$.

\begin{theorem}[\cite{Gray2006}]
Any circulant matrix $\bH$ can be diagonalized in the Fourier basis $\bF_N$: the eigenvectors of $\bH$ are given by the columns $\{\bof_m\}_{m\in \{0,\ldots,N-1\}}$ of the inverse DFT matrix $\bF_N^H$, and the associated eigenvalues $\{\lambda_m\}_{m\in \{0,\ldots,N-1\}}$ are given by $\sqrt{N} \hat{\bh}$, i.e. by the DFT of the first column $\bh$ of $\bH$: 
\begin{align*}
\forall m &\in \{0,\ldots,N-1\}, \quad \bH \bof_m= \lambda_m \bof_m\\
&\mbox{ where } \quad \lambda_m = \sqrt{N} \hat{h}_m=\sum_{n=0}^{N-1} h_n e^{ - \frac{j 2 \pi \: m \: n}{N}}. 
\end{align*}
Equivalently, in the Fourier domain, the circular convolution matrix $\bH$ becomes a diagonal matrix $\hat{\bH}= \diag\{\sqrt{N} \hat{\bh} \}.$
\end{theorem}
%
\begin{cor}
Consider the circular convolution $\by=\bH \bx$ \junk{\rm
  (\ref{eq:MatCircConv}) } of $\bx$ and $\by$. Let $\hat{\bx}=\bF_N
\bx$ and $\hat{\bh}= \bF_N \bh$ denote the normalized DFT of $\bx$ and
$\bh$. In the Fourier domain, the circular convolution becomes a
simple entry-wise multiplication of the components of $\sqrt{N}
\hat{\bh}$ with the components of $\hat{\bx}$: 
$\hat{\by}
= \bF_N \:  \by = \hat{\bH} \: \hat{\bx}$.
\end{cor}


\subsection{Privacy Model}

\subsubsection{Differential Privacy}

Two real-valued input vectors $\vec{x}, \vec{x}' \in [0, 1]^N$ are
\junk{\item \emph{$\ell_1$ neighbors}} \emph{neighbors} when
$\|\vec{x} - \vec{x}'\|_1 \leq 1$.


\begin{definition}
  A randomized algorithm $\alg$ satisfies $(\eps,
  \delta)$-differential privacy if for all neighbors $\vec{x},
  \vec{x}' \in [0, 1]^n$, and all measurable subsets $T$ of the
  support of $\alg$, we have
  \begin{equation*}
    \Pr[\alg(\vec{x}) \in T] \leq e^\eps\Pr[\alg(\vec{x}') \in T] + \delta,
  \end{equation*}
  where probabilities are taken over the randomness of $\alg$.
\end{definition}

\junk{Any algorithm which is $(\eps, \delta)$-differentially private for
$\ell_2$ neighbors is also $(\eps, \delta)$-differentially private for
$\ell_1$ neighbors.}

\subsubsection{Laplace Noise Mechanism}

\junk{In~\cite{DMNS}, Dwork et al.~proposed a generic method to construct
algorithms that approximate functions $f: [0, 1]^N \rightarrow X$,
where $(X, d)$ is a metric space. We will use an instantiation of
their method when $X$ is a subset of $\C^m$ and $d$ is the $\ell_p$
norm.}

\begin{definition}
  A function $f:[0, 1]^N \rightarrow \C$ has sensitivity $s$ if $s$
  is the smallest number such that for any two neighbors $\vec{x},
  \vec{x}' \in [0, 1]^N$,
  \begin{equation*}
    |f(\vec{x}) - f(\vec{x}')| \leq s.
  \end{equation*}
\end{definition}

\junk{\begin{theorem}[\cite{DMNS}, Theorem 2]\label{th:eps0-DiffPriv}
  Let $f:[0, 1]^N \rightarrow X$ ($X \subseteq \C^m)$ have
  $\ell_p$-sensitivity $s$ for $\ell_1$ \junk{(resp.~$\ell_2$)}
  neighbors. Suppose that on input $\vec{x}$, algorithm $\alg$ outputs
  $f(\vec{x}) + \vec{z}$, where $\vec{z}$ is sampled from $X$ with
  probability density $\alpha_{p, \eps, s}(X)\mu_{p, \eps,
    s}(\vec{z})$, where
  \begin{align}
    \mu_{p, \eps, s}(\vec{z}) &\triangleq \exp(-\eps \|\vec{z}\|_p/s)\\
    \alpha_{p, \eps, s}(X) &\triangleq \left(\int_{X}{\mu_{p, \eps,
          s}(\vec{z})d\vec{z}}\right)^{-1}
  \end{align}
  Then $\alg$ satisfies $(\eps, 0)$-differential privacy for $\ell_1$
  \junk{(resp.~$\ell_2$)} neighbors.
\end{theorem}

Sampling $\vec{z}$ from $\R^m$ with density $\alpha_{1, \eps,
  s}(\R^m)\mu_{1, \eps, s}$ is equivalent to sampling from the
$m$-dimensional Laplace distribution $\Lap(0, s/\eps)^m(x)=$.}

\begin{theorem}[\cite{DMNS}]\label{th:eps0-DiffPriv}
  Let $f:[0, 1]^N \rightarrow \C$ have sensitivity $s$. Suppose that
  on input $\vec{x}$, algorithm $\alg$ outputs $f(\vec{x}) + {z}$,
  where ${z} \sim \Lap(0, s/\eps)$. Then $\alg$ satisfies
  $(\eps, 0)$-differential privacy.
\end{theorem}

\subsubsection{Composition Theorems}
An important feature of differential privacy is its robustness: when
an algorithm is a ``composition'' of several differentially private
algorithms, the algorithm itself also satisfies differential privacy
constraints, with the privacy parameters degrading smoothly. The
results in this subsection quantify how the privacy parameters
degrade.

The first composition theorem is an easy consequence of the definition
of differential privacy:
\begin{theorem}[\cite{DMNS}]
  \label{thm:simple-composition}
  Let $\alg_1$ satisfy $(\eps_1, \delta_1)$-differential privacy and
  $\alg_2$ satisfy $(\eps_2, \delta_2)$-differential privacy, where
  $\alg_2$ could take the output of $\alg_1$ as input. Then the
  algorithm which on input $\vec{x}$ outputs the tuple
  $(\alg_1(\vec{x}), \alg_2(\alg_1(\vec{x}), \vec{x}))$ satisfies
  $(\eps_1 + \eps_2, \delta_1 + \delta_2)$-differential privacy.
\end{theorem}

In a more recent paper, Dwork et al.~proved a more sophisticated
composition theorem, which often gives asymptotically better bounds on
the privacy parameters. Next we state their theorem.
\begin{theorem}[\cite{dwork2010boosting}]\label{thm:fancy-comp}
  Let $\alg_1$, $\ldots$, $\alg_k$ be such that algorithm $\alg_i$
  satisfies $(\eps_i, 0)$-differential privacy. Then the algorithm
  that on input $\vec{x}$ outputs the tuple $(\alg_1(\vec{x})$, $\ldots$,
  $\alg_k(\vec{x}))$ satisfies $(\eps, \delta)$-differential privacy
  for any $\delta > 0$ and $$\eps \geq \sqrt{2\ln
    \left(\frac{1}{\delta}\right)\sum_{i = 1}^m{\eps_i^2}}.$$
\end{theorem}

\subsection{Accuracy}

In this paper we are interested in differentially private algorithms
for the \emph{convolution problem}. In the convolution problem, we are
given a \emph{public} sequence $h = \{h_1, \ldots, h_N\}$ and a
\emph{private} sequence $x = \{x_1, \ldots, x_N\}$. Our goal is to
design an algorithm $\alg$ that is $(\eps, \delta)$-differentially
private with respect to the private input $x$ (taken as column vector
$\vec{x}$), and approximates the convolution $h \ast x$. More
precisely,

\begin{definition}
  Given a vector $\vec{h} \in \R^N$ which defines a convolution matrix
  $\vec{H}$, the mean (expected) squared error ($\MSE$) of an algorithm $\alg$ is
  defined as
  \begin{equation*}
    \MSE = \sup_{\vec{x} \in \R^N}\frac{1}{N}\E[\|\alg(\vec{x})
    - \vec{Hx}\|_2^2].
  \end{equation*}
\end{definition}

Note that $\MSE$ measures the mean expected squared error
\emph{per output component}.

\junk{
\section{Basic Mechanisms}

\subsection{Exponential Noise Mechanism}

In~\cite{DMNS}, Dwork et al.~proposed a generic method to construct
algorithms that approximate functions $f: [0, 1]^N \rightarrow X$,
where $(X, d)$ is a metric space. We will use an instantiation of
their method when $X$ is a subset of $\C^m$ and $d$ is the $\ell_p$
norm.

We first need to introduce some notation. We denote the real $m$-sphere of radius $r$  as $r\R S^m = \{\vec{x} \in \R^{m+1}: \|\vec{x}\|_2 = r\}$ and the complex $m$-sphere of radius $r$ as $r\C S^m = \{\vec{x} \in \C^{m+1}: \|\vec{x}\|_2 = r\}$.
We have the following facts about the $m$-dimensional volume of $r \R S^m$ and $r \C
S^m$:
\begin{align}
  \int_{r\R S^m}{1 d\vec{x}} &= \frac{2\pi^{\frac{m + 1}{2}}}{\Gamma(\frac{m+1}{2})} \:\: r^m\\
  \int_{r\C S^m}{1 d\vec{x}} &=  \frac{2\pi^{m + 1}}{\Gamma(m
    + 1)} \:\: r^{2m+1} \\
\end{align}
Note that $r\C S^m$ can be thought of  as $r\R S^{2m+1}$.

We will also need some common distributions. The normal distribution
with mean $a$ and variance $b$ is denoted as $N(a, b)$. The
chi-squared distribution with $m$ degrees of freedom is denoted as
$\chi^2_m$. It is distributed as the sum of $m$ squared samples from
$N(0, 1)$ and therefore has mean $m$. The gamma distribution with
scale $b > 0$ and shape $a > 0$ is denoted $\Gamma(a, b)$. Its
probability density function (defined for $x \geq 0$) is
\begin{equation*}
  \gamma_{a, b}(x) = \frac{1}{\Gamma(a)b^a}x^{a-1}\exp(-\frac{x}{b}).
\end{equation*}
The mean of the $\Gamma(a, b)$ is $ab$ and its variance is $ab^2$.

\begin{definition}
  A function $f:[0, 1]^N \rightarrow \C^m$ has $\ell_p$-sensitivity
  $s$ for $\ell_1$ \junk{(resp.~$\ell_2$)} neighbors, if $s$ is the smallest
  number such that for any two $\ell_1$ \junk{(resp.~$\ell_2$)} neighbors
  $\vec{x}, \vec{x}' \in [0, 1]^N$,
  \begin{equation*}
    \|f(\vec{x}) - f(\vec{x}')\|_p \leq s.
  \end{equation*}
\end{definition}

\begin{theorem}[\cite{DMNS}, Theorem 2]\label{th:eps0-DiffPriv}
  Let $f:[0, 1]^N \rightarrow X$ ($X \subseteq \C^m)$ have
  $\ell_p$-sensitivity $s$ for $\ell_1$ \junk{(resp.~$\ell_2$)}
  neighbors. Suppose that on input $\vec{x}$, algorithm $\alg$ outputs
  $f(\vec{x}) + \vec{z}$, where $\vec{z}$ is sampled from $X$ with
  probability density $\alpha_{p, \eps, s}(X)\mu_{p, \eps,
    s}(\vec{z})$, where
  \begin{align}
    \mu_{p, \eps, s}(\vec{z}) &\triangleq \exp(-\eps \|\vec{z}\|_p/s)\\
    \alpha_{p, \eps, s}(X) &\triangleq \left(\int_{X}{\mu_{p, \eps,
          s}(\vec{z})d\vec{z}}\right)^{-1}
  \end{align}
  Then $\alg$ satisfies $(\eps, 0)$-differential privacy for $\ell_1$
  \junk{(resp.~$\ell_2$)} neighbors.
\end{theorem}

Sampling $\vec{z}$ from $\R^m$ with density $\alpha_{1, \eps,
  s}(\R^m)\mu_{1, \eps, s}$ is equivalent to sampling from the
$m$-dimensional Laplace distribution $\Lap(0, s/\eps)^m$. Next we show
$\mu_{2, \eps, s}$ is closely related to the gamma distribution.

\begin{lemma}\label{lm:l2-sample}
  Let $\vec{z}$ be sampled as follows:
  \begin{itemize}
  \item pick $r$ from the distribution $\Gamma(2m, s/\eps)$
    (resp. $\Gamma(m, s/\eps)$)
  \item pick $\vec{z}$ uniformly from the complex sphere $r\C S^{m-1}$
    (resp. from the real sphere $ r \R S^{m-1}$)
  \end{itemize}
  Then the distribution of $\vec{z}$ has probability density
  $\alpha_{2, \eps, s}(\C^m)\mu_{2,
    \eps, s}$ (resp. $\alpha_{2, \eps, s}(\R^m)\mu_{2, \eps, s}$).
\end{lemma}
\begin{proof}
  We will prove the complex case, the real case is analogous.

  Let $p$ be the probability density of $\vec{z}$ sampled as in the
  statement of the lemma. It is enough to show that
  \begin{equation}\label{eq:ratio-dens}
    \forall \vec{z}, \vec{z}' \in \C^m: \frac{p(\vec{z})}{p(\vec{z}')}
    = \frac{\mu_{2, \eps,  s}(\vec{z})}{\mu_{2, \eps, s}(\vec{z}')}.
  \end{equation}
  Let $\gamma = \gamma_{2m, s/\eps}$. We have
  \begin{equation}\label{eq:sampl-dens}
    p(\vec{z}) = \gamma(\|\vec{z}\|_2)\left(\int_{\|\vec{z}\|_2
        \C S^{n-1}}{1}\right)^{-1} =  c(m, \eps, s)
    \frac{\|\vec{z}\|_2^{2m-1}\exp(-\eps\|\vec{z}\|_2/s)}{\|\vec{z}\|_2^{2m-1}}
    = c(m, \eps, s)\exp(-\eps\|\vec{z}\|_2/s),
  \end{equation}
  where $c(m, \eps, s)=\frac{\Gamma(m)}{2 \pi^m \Gamma(2m)} \left(\frac{\eps}{s}\right)^2$ is a constant that depends only on $m$, $\eps$ and $s$   (but not on $\vec{z}$). Equation (\ref{eq:ratio-dens}) follows
  immediately from (\ref{eq:sampl-dens}).
\end{proof}

\begin{cor}
  Let $\vec{z}$ be sampled from $\C^m$ with probability density
  $\alpha_{2, \eps, s}(\C^m)\mu_{2, \eps, s}(\vec{z})$. Then
  $\E[\|\vec{z}\|_2] = \frac{2ms}{\eps}$.
\end{cor}

\junk{
Lemma~\ref{lm:l2-sample} allows us to derive the covariance matrix of
$\vec{z}$ when $\vec{z}$ is sampled from $\C^m$ with density
$\alpha_{2, \eps, s}(\C^m)\mu_{2, \eps, s}(\vec{z})$. }

\begin{lemma}
  Let $\vec{w}$ be sampled uniformly at random from the real sphere
  $\R S^{m-1}$. Then
  \begin{equation*}
    \E[\vec{w}\vec{w}^T] = \frac{1}{m}\bI_m.
  \end{equation*}
\end{lemma}
\begin{proof}
  Let $\vec{y} \in \C^{m}$ be sampled from $N(0, 1)^m$. The distribution of
  $\vec{y}$ is rotation invariant and $\|\vec{y}\|_2^2$ is distributed
  as a sample from $\chi^2_m$. Therefore, $\vec{y}$ can be
  equivalently sampled by sampling $r^2$ from $\chi^2_m$, independently
  sampling $\vec{w}$ uniformly from $S^{m-1}$, and letting $\vec{y} =
  r\vec{w}$. It follows that for any $i$, $j$, $\E[y_iy_j] =
  \E[r^2 w_iw_j] = \E[\|\vec{y}\|_2^2w_iw_j]$. Therefore,
  \begin{equation}
    \E[\vec{w}\vec{w}^T] =
    \frac{1}{\E[\|\vec{y}\|_2^2]}\E[\vec{y}\vec{y}^T] = \frac{1}{m}\bI_m.
  \end{equation}
\end{proof}

\begin{lemma}\label{lm:l2-covar}
  Let $\vec{z}$ be sampled from $\R^m$ with density
  $\alpha_{2, \eps, s}(\R^m)\mu_{2, \eps, s}(\vec{z})$. Then
  \begin{equation}
    \E[\vec{z}\vec{z}^T] = \Theta\left(\frac{ms^2}{\eps^2}\right)\bI_m.
  \end{equation}
\end{lemma}
\begin{proof}
  Let $r$ be sampled from $\Gamma(m, \frac{s}{\eps})$  and $\vec{w}$ sampled
  uniformly from $\R S^m$. By Lemma~\ref{lm:l2-sample}, $\vec{z}$ is
  distributed identically to $r\vec{w}$. Therefore,
  \begin{equation}
    \E[\vec{z}\vec{z}^T] = \E[r^2]\E[\vec{w}\vec{w}^T] =
    m^2\frac{s^2}{\eps^2}\left(1 + \frac{1}{m}\right)\frac{1}{m}\bI_m =
    \Theta\left(\frac{ms^2}{\eps^2}\right)\bI_m.
  \end{equation}
\end{proof}

\junk{
\subsection{Gaussian Noise Mechanisms}

\begin{theorem}[OurDataOurSelves2006]\label{th:eps-delta-DiffPriv}

\end{theorem}
}

\subsection{Composition}

An important feature of differential privacy is its robustness: when
an algorithm is a ``composition'' of several differentially private
algorithms, the algorithm itself also satisfies differential privacy
constraints, with the privacy parameters degrading smoothly. The
results in this subsection quantify how the privacy parameters
degrade. In the following theorems, if the component algorithms
satisfy -differential privacy for $\ell_1$ (resp.~$\ell_2$) neighbors,
the composition satisfies differential privacy for $\ell_1$
(resp.~$\ell_2$ neighbors).

The following composition theorem is an easy consequence of the
definition of differential privacy:
\begin{theorem}[\cite{DMNS}]
  \label{thm:simple-composition}
  Let $\alg_1$ satisfy $(\eps_1, \delta_1)$-differential privacy and $\alg_2$ satisfy $(\eps_2,
  \delta_2)$-differential privacy, where $\alg_2$ could take the
    output of $\alg_1$ as input. Then the algorithm which on input
  $\vec{x}$ outputs the tuple $(\alg_1(\vec{x}),
  \alg_2(\alg_1(\vec{x}), \vec{x}))$ satisfies $(\eps_1 + \eps_2,
  \delta_1 + \delta_2)$-differential privacy.
\end{theorem}

In a more recent paper, Dwork et al.~proved a more sophisticated
composition theorem, which often gives asymptotically better bounds on
the privacy parameters. Next we state their theorem.
\begin{theorem}[\cite{dwork2010boosting}]\label{thm:fancy-comp}
  Let $\alg_1$, $\ldots$, $\alg_k$ be such that algorithm $\alg_i$
  satisfies $(\eps_i, 0)$-differential privacy. Then, for any $\delta
  > 0$, the algorithm that on input $\vec{x}$ outputs the tuple
  $(\alg_1(\vec{x}), \ldots, \alg_k(\vec{x}))$ satisfies $(\eps,
  \delta)$-differential privacy for any $\eps \geq \sqrt{2\ln
    \left(\frac{1}{\delta}\right)\sum_{i = 1}^m{\eps_i^2}}$.
\end{theorem}
} 

\section{Lower Bounds}\label{sec:LB}

In this section we derive a spectral lower bound on mean squared error of
differentially private approximation algorithms for circular
convolution. We prove that this bound is nearly tight for every
fixed $\vec{h}$ in the following section. The lower bound is state as
Theorem~\ref{thm:main-lb}. 

\begin{theorem}
  \label{thm:main-lb}
  Let $\vec{h} \in \R^N$ be an arbitrary real vector and let us
  relabel the Fourier coefficients of $\vec{h}$ so that $|\hat{h}_0|
  \geq \ldots \geq |\hat{h}_{N-1}|$. For all sufficiently
  small $\eps$ and $\delta$, the expected mean squared error
  $\MSE$ of any $(\eps, \delta)$-differentially private algorithm
  $\alg$ that approximates $\vec{h} \ast \vec{x}$ is at least
  \begin{equation}
    \label{eq:spec-lb}
    \MSE = \Omega\left(\max_{K =
        1}^{N}{\frac{K^2\hat{h}^2_{K-1}}{N\log^2N}}\right).
  \end{equation}
\end{theorem}
For the remainder of the paper, we define the notation
$\specLB(\vec{h})$ for the right hand side of (\ref{eq:spec-lb}), i.e.
  $\specLB(\vec{h}) = \max_{K = 1}^N{\frac{K^2\hat{h}^2_{K-1}}{N\log^2N}}$.

The proof of Theorem~\ref{thm:main-lb} is based on recent
work~\cite{stoc} connecting combinatorial discrepancy and
privacy. Adapting a strategy due to Hardt and
Talwar~\cite{hardt2010geometry}, we instantiate the basic discrepancy
lower bound for any matrix $\vec{PA}$, where $\vec{P}$ is a projection
matrix, and use the maximum of these lower bounds. However, we need to
resolve several issues that arise in the setting of $(\eps,
\delta)$-differential privacy. While projection works naturally with
the volume-based lower bounds of Hardt and Talwar, the connection
between the discrepancy of $\vec{A}$ and $\vec{PA}$ is not immediate,
since discrepancy is a combinatorially defined quantity. Our main
technical contribution in this section is analyzing the discrepancy of
$\vec{PA}$ via the determinant lower bound of Lov\'{a}sz, Spencer,
Vesztergombi. This approach was generalized and extended by Nikolov,
Talwar, and Zhang~\cite{geometry} to show nearly optimal lower bounds
for arbitrary linear maps.

We start our presentation with preliminaries from prior work and then
we develop our lower bounds for convolutions. 

\subsection{Discrepancy Preliminaries}

We define ($\ell_2$) hereditary discrepancy as
\begin{equation*}
  \herdisc(\vec{A}) = \max_{W \subseteq [N]}{\min_{\vec{v} \in \{-1,
      +1\}^W}{\|\vec{Av}\|_2}}.
\end{equation*}

The following result connects discrepancy and differential privacy:
\begin{theorem}[\cite{stoc}]
  \label{thm:disc-lb}
  Let $\vec{A}$ be an $M \times N$ complex matrix and let $\alg$ be an
  $(\eps, \delta)$-differentially private algorithm for sufficiently
  small constant $\eps$ and $\delta$. There exists a constant $C$ and a
  vector $\vec{x} \in \{0, 1\}^N$ such that
    $\E[\|\alg(\vec{x}) - \vec{Ax}\|_2^2] \geq
    C\frac{\herdisc(\vec{A})^2}{\log^2 N}$.
\end{theorem}

The determinant lower bound for hereditary discrepancy due to
Lov\'{a}sz, Spencer, and Vesztergombi gives us a spectral lower bound
on the noise required for privacy.

\begin{theorem}[\cite{lovasz1986discrepancy}]
  There exists a constant $C'$ such that for any complex $M \times N$
  matrix $\vec{A}$,
    $\herdisc(\vec{A}) \geq C'\max_{K, \vec{B}}{\sqrt{K}|\det(\vec{B})|^{1/K}}$,
  where $K$ ranges over $[\min\{M, N\}]$ and $\vec{B}$ ranges
  over $K \times K$ submatrices of $\vec{A}$.
\end{theorem}

\begin{corollary}
  \label{cor:det-lb}
  Let $\vec{A}$ be an $M \times N$ complex matrix and let $\alg$ be an
  $(\eps, \delta)$-differentially private algorithm for sufficiently
  small constant $\eps$ and $\delta$. There exists a constant $C$ and a
  vector $\vec{x} \in \{0, 1\}^N$ such that, for any $K \times K$
  submatrix $\vec{B}$ of $\vec{A}$,
    $\E[\|\alg(\vec{x}) - \vec{Ax}\|_2^2] \geq
    C\frac{K|\det(\vec{B})|^{2/K}}{\log^2 N}$.
\end{corollary}

\subsection{Proof of Theorem~\ref{thm:main-lb}}

We exploit the power of the determinant lower bound of
Corollary~\ref{cor:det-lb} by combining the simple but very useful
observation that projections do not increase mean squared error with a
lower bound on the maximum determinant of a submatrices of a
rectangular matrix. We present these two ingredients in sequence and
finish the section with a proof of Theorem~\ref{thm:main-lb}. 

\begin{lemma}
  \label{lm:subsp-det-lb}
  Let $\vec{A}$ be an $M \times N$ complex matrix and let $\alg$ be an
  $(\eps, \delta)$-differentially private algorithm for sufficiently
  small constant $\eps$ and $\delta$. There exists a constant $C$ and
  a vector $\vec{x} \in \{0, 1\}^N$ such that for any $L \times M$
  projection matrix $\vec{P}$ and for any $K \times K$ submatrix
  $\vec{B}$ of $\vec{PA}$,
    $\E[\|\alg(\vec{x}) - \vec{Ax}\|_2^2] \geq
    C\frac{K|\det(\vec{B})|^{2/K}}{\log^2 N}$.
\end{lemma}
\begin{proof}
  We show that there exists an
  $(\eps, \delta)$-differentially private algorithm $\mathcal{B}$
  that satisfies
  \begin{equation}
    \label{eq:proj-alg}
    \E[\|\mathcal{B}(\vec{x}) - \vec{PAx}\|_2^2] \leq
    \E[\|\alg(\vec{x}) - \vec{Ax}\|_2^2].
  \end{equation}
  Then we can apply Corollary~\ref{cor:det-lb} to $\mathcal{B}$ and
  $\vec{PA}$ to prove the corollary.

  The algorithm $\mathcal{B}$ on input $\vec{x}$ outputs $\vec{Py}$
  where $\vec{y} = \alg(\vec{x})$. Since $\mathcal{B}$ is a function
  of $\alg(\vec{x})$ only, it satisfies $(\eps, \delta)$-differential
  privacy by Theorem~\ref{thm:simple-composition}. It satisfies
  (\ref{eq:proj-alg}) since for any $\vec{y}$ and any projection
  matrix $\vec{P}$ it holds that $\|\vec{P}(\vec{y - Ax})\|_2 \leq
  \|\vec{y - Ax}\|_2$.
\end{proof}

Our main technical tool is a linear algebraic fact connecting the
determinant lower bound for $\vec{A}$ and the determinant lower bound
for any projection of $\vec{A}$.

\begin{lemma}
  \label{lm:subsp-det}
  Let $\vec{A}$ be an $M \times N$ complex matrix with singular values
  $\lambda_1 \geq \ldots \geq \lambda_N$ and let $\vec{P}$ be a
  projection matrix onto the span of the left singular vectors
  corresponding to $\lambda_1, \ldots, \lambda_K$. There exists a
  constant $C$ and $K \times K$ submatrix $\vec{B}$ of $\vec{PA}$ such that
  \begin{equation*}
    |\det(\vec{B})|^{1/K} \geq C \sqrt{\frac{K}{N}} \left(\prod_{i =
        1}^K{\lambda_i}\right)^{1/K}
  \end{equation*}
\end{lemma}
\begin{proof}
  Let $\vec{C} = \vec{PA}$ and consider the matrix $\vec{D} = \vec{C}\vec{C}^H$. It has eigenvalues
  $\lambda_1^2, \ldots, \lambda_K^2$, and therefore 
  \[
  \det(\vec{D}) = \prod_{i =  1}^K{\lambda^2_i}.
  \]
  On the other hand, by the Binet-Cauchy formula
  for the determinant, we have
  \begin{align*}
    \det(\vec{D}) &= \det(\vec{C}\vec{C}^H) \\
    &= \sum_{S \in {[N]\choose K}}{\det(\vec{C}|_S)^2} \\&
    \leq {N \choose K}\max_{S \in {[N]\choose   K}}{\det(\vec{C}|_S)^2}.
  \end{align*}
  Rearranging and raising to the power $1/2K$, we get that there
  exists a $K \times K$ submatrix of $\vec{C}$ such that
  \begin{equation*}
    |\det(\vec{B})|^{1/K} \geq {N \choose K}^{-1/2K}\left(\prod_{i =
        1}^K{\lambda_i}\right)^{1/K}.
  \end{equation*}
  Using the bound ${N \choose K} \leq \left(\frac{Ne}{K}\right)^K$
  completes the proof.
\end{proof}

We can now prove our main lower bound theorem by combining
Lemma~\ref{lm:subsp-det-lb} and Lemma~\ref{lm:subsp-det}.

\begin{proof}[of Theorem~\ref{thm:main-lb}]
  As usual, we will express $\vec{h} \ast \vec{x}$ as the linear map
  $\vec{Hx}$, where $\vec{H}$ is the convolution matrix for
  $\vec{h}$. By Lemma~\ref{lm:subsp-det-lb}, it suffices to show
  that for each $K$, there exists a projection matrix
  $\vec{P}$ and a $K \times K$ submatrix $\vec{B}$ of $\vec{PH}$ such
  that $|\det(B)|^{1/K} \geq \Omega(\sqrt{K}|\hat{h}_K|)$. Recall
  that the eigenvalues of $\vec{H}$ are $\sqrt{N}\hat{h}_0, \ldots,
  \sqrt{N}\hat{h}_{N-1}$, and, therefore, the $i$-th singular value of
  $\vec{H}$ is $\sqrt{N}|\hat{h}_{i-1}|$. By Lemma~\ref{lm:subsp-det},
  there exists a constant $C$, a projection matrix $P$, and a submatrix $\vec{B}$
  of $\vec{PH}$ such that
  \begin{equation*}
    |\det(\vec{B})|^{1/K} \geq C\sqrt{\frac{K}{N}} \left(\prod_{i =
        0}^{K-1}{\sqrt{N}|\hat{h}_i|}\right)^{1/K} \geq
    C\sqrt{K}|\hat{h}_K|.
  \end{equation*}
  This completes the proof.
\end{proof}

\section{Upperbounds}\label{sec:UB}

Standard $(\eps, \delta)$-privacy techniques such as input perturbation or output
perturbation in the time or in the frequency domain lead to mean
squared error, at best, proportional to $\|\vec{h}\|_2^2$. 

Next we describe an algorithm which is nearly optimal for $(\eps,
\delta)$-differential privacy. This algorithm is derived by
formulating the error of a natural class of private algorithms as a
convex program and finding a closed form solution. An alternative
solution that partitions the spectrum of $\vec{H}$ geometrically is
described in Appendix~\ref{app:specpart}. The class of algorithms we
consider is those which add independent Gaussian
noise to the Fourier coefficients of the private input
$\vec{x}$. Interestingly, we show that this simple strategy is nearly
optimal for computing convolution maps. 




Consider the class of algorithms, which first add independent
Laplacian noise variables $z_i=\Lap(0,b_i)$ to the Fourier coefficients $\hat{x}_i$ to compute
$\tilde{x}_i = \hat{x}_i + z_i$,
and then output
$\tilde{\by}= \bF_N^H\hat{\bH}\tilde{\bx}$.
This class of algorithms is parameterized by the vector
$\bb=(b_0,\ldots,b_{N-1})$; a member of the class will be denoted
$\alg(\bb)$ in the sequel.  The question we address is: For given
$\eps,\delta>0$, how should the noise parameters $\bb$ be chosen such
that the algorithm $\alg(\bb)$ achieves $(\eps,\delta)$-differential
privacy in $\bx$ for $\ell_1$ neighbors, while minimizing the mean
squared error $\MSE$? It turns out that by convex programming duality we can
derive a closed form expression for the optimal $\bb$, and moreover,
the optimal $\alg(\bb)$ is nearly optimal \emph{among all $(\eps,
  \delta)$-differentially private algorithms}. The optimal parameters
are used in Algorithm~\ref{alg:CP}.


\begin{algorithm}[t]
  \caption{\textsc{Fourier Mechanism}}\label{alg:CP}
  {\fontsize{9}{9}\selectfont
  \begin{algorithmic}
    \STATE Set $\gamma = \frac{2\ln(1/\delta) \| \vec{\hat{h}} \|_1}{\eps^2 N }$
    \STATE Compute $\vec{\hat{x}} = \bF_N\vec{x}$ and $\vec{\hat{h}} =
    \bF_N \vec{x}$.
    \FORALL{$i \in \{0,\ldots, N-1\}$}
    \IF{$|\hat{h}_i| > 0$}
        \STATE Set $z_i =\Lap\left(\sqrt{\frac{\gamma}{|\hat{h}_i|}}\right)$
    \ELSIF{$|\hat{h}_i| = 0$}
        \STATE Set $z_i = 0$
    \ENDIF
    \STATE Set $\tilde{x}_i = \hat{x}_i + z_i$.
    \STATE Set $\bar{y}_i = \sqrt{N}\hat{h}_i \tilde{x}_i$.
    \ENDFOR
    \STATE Output $\vec{\tilde{y}} =  \bF_N^H \vec{\bar{y}}$
  \end{algorithmic}
}
\end{algorithm}


\begin{theorem}
  \label{thm:CP-error}
  Algorithm~\ref{alg:CP} satisfies $(\eps,\delta)$-differential privacy, and  achieves expected mean squared error
  \begin{equation}\label{eq:accuracy}
    \MSE = 4  \frac{\ln(1/\delta)}{\eps^2 N}\| \vec{\hat{h}} \|_1^2.
  \end{equation}
  Moreover, Algorithm~\ref{alg:CP} runs in time $O(N \log N)$.
\end{theorem}

Before proving Theorem~\ref{thm:CP-error}, we show that it implies
that Algorithm~\ref{alg:CP} is almost optimal \emph{for any given
  $\vec{h}$}. 

\begin{theorem}\label{thm:CP-opt}
  For any $\vec{h}$, Algorithm ~\ref{alg:CP} satisfies $(\eps,
  \delta)$-differential privacy and achieves expected mean squared
  error $O\left(\specLB(\vec{h})\frac{\log^2{N} \log^2{|I|} \ln (1/\delta)}{\eps^2}\right)$.
\end{theorem}
\begin{proof}
Assume that $|\hat{h}_0|>|\hat{h}_1| >\ldots>|\hat{h}_{N-1}|$. Then, by definition of $I=\{0\leq i \leq N-1 : |\hat{h}_i|>0 \}$, we have  $|\hat{h}_j|=0$, for all $j>|I|-1$. Thus,
\begin{align}
\|\hat{\bh}\|_1 &= \sum_{i=0}^{ |I|-1} |\hat{h}_i| 
= \sum_{i=1}^{|I|} \frac{1}{i} i |\hat{h}_{i-1}| \notag\\
&\leq \left(\sum_{i=1}^{|I|} \frac{1}{i}\right) \sqrt{N} \log N \sqrt{\specLB(\vec{h})}\notag\\
&= H_{|I|} \sqrt{N} \log N \sqrt{\specLB(\vec{h})}, \label{eq:h1}
\end{align}
where $H_m=\sum_{i=1}^{m} \frac{1}{i}$ denotes the $m$-th harmonic number. Recalling that $H_m= O (\log m)$, and combining the bound (\ref{eq:h1}) with the expression of the MSE (\ref{eq:accuracy}) yields the desired bound.
\end{proof}

\begin{proof}[of Theorem~\ref{thm:CP-error}]
For running time, we note that our algorithm is no more expensive than
computing a Fast Fourier Transform, which can be done in  $O(N
\log N)$ arithmetic operations using the classical Cooley-Tukey algorithm, for
example. 

Denote the set $I=\{0\leq i \leq N-1 : |\hat{h}_i|>0 \}$. We formulate the problem of finding the algorithm $\alg(\bb)$ which minimizes MSE subject to privacy constraints as the following optimization problem:
\begin{align}
\min_{\{b_i\}_{i\in I}} & \sum_{i \in I} b_i^2 |\hat{h}_i|^2 \label{eq:obj}\\
\mbox{s.t.} &  \sum_{i \in I} \frac{1}{N b_i^2} = \frac{\eps^2}{2 \ln (1/\delta)}\label{eq:privCons}\\
&  b_i > 0, \forall i \in I.\label{eq:non-neg}
\end{align}
Next we justify this formulation.

\textbf{Privacy Constraint.}
We first show that the output $\tilde{\by}$ of an algorithm $\alg(\bb)$ is an $(\eps, \delta)$-differentially private function of $\vec{x}$, if the constraint (\ref{eq:privCons}) is satisfied.
%
Denote $\vec{\bar{y}}= \hat{\bH} \tilde{\bx}$. If $\vec{\bar{y}}$ is an $(\eps,\delta)$-differentially private function of $\vec{x}$, then by Theorem~\ref{thm:simple-composition}, $\vec{\tilde{y}}$ is also $(\eps,\delta)$-differentially private, since the computation of $\vec{\tilde{y}}$ depends only on $\bF_N^H$ and $\vec{\bar{y}}$ and not on $\vec{x}$ directly.
Thus we can focus on the requirements on $\bb$ for which $\vec{\bar{y}}$  is $(\eps,\delta)$ private.

If $i \notin I$, then $\bar{y}_i=0$ and does not affect privacy
regardless of $b_i$. Thus, we can set $b_i=0$ for all $i\notin I$.
If $i \in I$, we first characterize the $\ell_1$-sensitivity of $\hat{x}_i$  as a function of $\vec{x}$. Recall that $\hat{x}_i=\vec{f}_i^H \vec{x}$ is the inner product of $\vec{x}$ with the Fourier basis vector $\vec{f}_i$. The sensitivity of $\hat{x}_i$ is therefore $\| f_i \|_{\infty} = \frac{1}{\sqrt{N}}$, $\forall i$. Then, by Theorem~\ref{th:eps0-DiffPriv}, $\tilde{x}_i = \hat{x}_i + \Lap\left(0,b_i\right)$ is $\eps_i$-differentially private in $\vec{x}$, with $\eps_i =\frac{1}{\sqrt{N} b_i} $. The computation of $\bar{y}_i$ depends only on $\hat{h}_i$ and $\tilde{x}_i$, thus, by Theorem~\ref{thm:simple-composition}, $\bar{y}_i$  is $\frac{1}{\sqrt{N} b_i}$-differentially private in $\vec{x}$.

Finally, by Theorem~\ref{thm:fancy-comp}, $\vec{\bar{y}}$ is $(\eps,\delta)$ differentially private for any $\delta > 0$, as long as\junk{
\begin{align*}
    \eps^2
    = 2\ln(1/\delta) \sum_{i=0}^{N-1} \eps_i^2
    = 2\ln(1/\delta)\sum_{i\in I}\frac{1}{N b_i^2},
\end{align*}}
 constraint (\ref{eq:privCons}) holds.

\textbf{Accuracy Objective.}
We show that finding the algorithm $\alg(\bb)$ which minimizes the $\MSE$ is equivalent to finding the parameters $b_i \geq 0$, $i\in I$, which minimize the objective function (\ref{eq:obj}).
Note that $\tilde{\by}=\bF_N^H \hat{\vec{H}} \tilde{\vec{x}}= \bF_N^H \hat{\vec{H}} (\bF_N \bx +\bz) =\by + \bF_N^H \hat{\vec{H}} \bz$.  Thus, the output $\tilde{\by}$ is unbiased: $\E[\tilde{\by}]= \by$.
 The mean squared error is given by:
  \begin{align*}
  \MSE
  &=\frac{1}{N} \E [\| \bF_N^H \hat{\vec{H}} \bz \|_2^2 ]\\
  &=\frac{1}{N} \E [\tr (\bF_N^H \hat{\vec{H}} \bz \bz^H \hat{\vec{H}}^H \bF_N)]\\
  &=\frac{1}{N}\tr (\hat{\vec{H}}^2 \E [\bz \bz^H] )
  = 2 \sum_{i\in I} |\hat{h}_i|^2 b_i^2,
  \end{align*}
which yields the desired objective function (\ref{eq:obj}).

\textbf{Closed Form Solution.}
\junk{Using (\ref{eq:obj}), and (\ref{eq:privCons}), the problem can be cast as the following optimization
\begin{equation}\label{pb:bCP}
\begin{split}
\min_{\{b_i\}_{i\in I}} & \sum_{i \in I} b_i^2 |\hat{h}_i|^2 \\
\mbox{s.t.} &  \sum_{i \in I} \frac{1}{N b_i^2} = \frac{\eps^2}{2 \ln (1/\delta)}\\
&  b_i > 0, \forall i \in I.\\
\end{split}
\end{equation}}
The  program (\ref{eq:obj})--(\ref{eq:non-neg}) is convex in
$1/b_i^2$. Using the KKT conditions of this program, we can derive a closed form optimal solution: $b_i^* = \sqrt{(2\ln(1/\delta)\|\hat{\vec{h}}\|_1)/(N\eps^2|\hat{h}_i|)}$ when $i \in I$ and $b_i^* = 0$ otherwise. Substituting these values back into the objective finishes the proof. Full details of the analysis of the convex program can be found in Appendix~\ref{app:KKT}.
\end{proof}

%
%

\junk{
\begin{proof}
 \textbf{Privacy.} We will show that $\vec{\bar{y}}$ is an $(\eps,
  \delta)$-differentially private function of $\vec{x}$. The computation of $\vec{\tilde{y}}$ depends only on $\bF_N^H$ and $\vec{\bar{y}}$ and not
  on $\vec{x}$ directly, so, by Theorem~\ref{thm:simple-composition},
  incurs no loss in privacy.

  Let $\gamma = \frac{2\ln(1/\delta) \| \vec{\hat{h}} \|_1}{\eps^2 N }$, and denote by $I=\{0\leq i \leq N-1 : |\hat{h}_i|>0 \}$. We first show that $\bar{y}_i=\sqrt{N}\hat{h}_i \tilde{x}_i$ is an $\eps_i$-differentially private function of $\bx$, where
  \begin{equation}
  \eps_i=
  \left\{ \begin{array}{cc}
            \sqrt{\frac{|\hat{h}_i|}{N \gamma}}, & \mbox{if } i\in I \\
            0, & \mbox{if } i \notin I.
          \end{array}
  \right.
  \end{equation}
  If $i \notin I$, then $\bar{y}_i=0$, which is a $0$-differentially private function of $\bx$.
  If $i \in I$, we first characterize the $\ell_1$-sensitivity of $\hat{x}_i$  as a function of $\vec{x}$. Recall that $\hat{x}_i=\vec{f}_i^H \vec{x}$ is the inner product of $\vec{x}$ with the Fourier basis vector $\vec{f}_i$. The sensitivity of $\hat{x}_i$ is therefore $\| f_i \|_{\infty} = \frac{1}{\sqrt{N}}$, $\forall i$. Then, by Theorem~\ref{th:eps0-DiffPriv}, $\tilde{x}_i = \hat{x}_i + \Lap\left(\sqrt{\frac{\gamma}{|\hat{h}_i|}}\right)$ is $\eps_i$-differentially private in $\vec{x}$, with $\eps_i = \sqrt{\frac{|\hat{h}_i|}{N \gamma}} $. The computation of $\bar{y}_i$ depends only on $\hat{h}_i$ and $\tilde{x}_i$, thus, by Theorem~\ref{thm:simple-composition}, $\bar{y}_i$  is $\sqrt{\frac{|\hat{h}_i|}{N \gamma}}$-differentially private in $\vec{x}$.

By Theorem~\ref{thm:fancy-comp}, $\vec{\bar{y}}$ is $(\eps',
  \delta)$ differentially private for any $\delta > 0$, where
  \begin{align*}
    \eps'^2 &= 2\ln(1/\delta) \sum_{i=0}^{N-1} \eps_i^2\\
    &= 2\ln(1/\delta)\sum_{i\in I}\frac{|\hat{h}_i|}{N \gamma}= \eps^2
  \end{align*}

 \textbf{Accuracy.}
  Note that $\tilde{\by}=\bF_N^H \hat{\vec{H}} \tilde{\vec{x}}= \bF_N^H \hat{\vec{H}} (\bF_N \bx +\bz) =\by + \bF_N^H \hat{\vec{H}} \bz$.
  Thus, the output $\tilde{\by}$ is unbiased: $\E[\tilde{\by}]= \by$.
  The per-component mean squared error is given by:
  \begin{equation}
  \begin{split}
  \MSE
  &=\frac{1}{N} \E [\| \bF_N^H \hat{\vec{H}} \bz \|_2^2 ]\\
  &=\frac{1}{N} \E [\tr (\bF_N^H \hat{\vec{H}} \bz \bz^H \hat{\vec{H}}^H \bF_N)]\\
  &=\frac{1}{N}\tr (\hat{\vec{H}}^2 \E [\bz \bz^H] )\\
  &= \frac{1}{N} \sum_{i\in I} N |\hat{h}|_i^2 2 \frac{\gamma}{|\hat{h}_i|}\\
  &=2 \gamma \| \hat{\bh} \|_1 = 4  \frac{\ln(1/\delta)}{\eps^2 N}\| \vec{\hat{h}} \|_1^2
  \end{split}
  \end{equation}

 \textbf{Class optimality}
We refer the reader to Appendix~\ref{ap:ClassOptim} for the the formulation of the problem as a convex optimization whose solution yields the parameters of Algorithm~$\ref{alg:CP}$.

\end{proof}
} 

\section{Generalizations and Applications}
\label{sect:gen-apps}

In this section we describe some generalizations and applications of
our lower bounds and algorithms for private convolution. 

\junk{We first
describe some applications of private circular convolution to problems
in finance. Then we define a standard generalization of circular
convolution to convolution over Abelian groups, and describe an
application of private convolution over $(\mathbb{Z}/2)^d$ to
computing database queries privately.}

\subsection{Compressible Convolutions}\label{sect:compr}

A case of special interest is convolutions $h \ast x$ where $h$ is a
compressible sequence. Such cases appear in practice in signal
processing. For compressible $h$ we can show
that Algorithm~\ref{alg:CP} outperforms input and output
perturbation. First we present a definition of compressible sequences
and then we give the improved upper bounds. A specific example of
private compressible convolutions is developed in
Section~\ref{sect:gen-marg} in the context of computing marginal
queries. 

\begin{definition}\label{def:compressible}
  A vector $\vec{h} \in \R^N$ is  $(c, p)$-compressible (in the Fourier basis) if it satisfies:
  \begin{equation*}
    \forall 0\leq i \leq N-1 : |\hat{h}_i|^2 \leq c\frac{1}{(i+1)^p}.
  \end{equation*}
\end{definition}

\begin{theorem}\label{thm:compressible}
  Let $\vec{h}$ be a $(c, p)$-compressible vector for some constant $p
  > 2$. Then Algorithm~\ref{alg:CP} satisfies $(\eps,
  \delta)$-differential privacy and achieves expected mean squared error
  $O\left(\frac{c^2 \log^2 N \log (1/\delta)}{N\eps^2}\right)$ for
  $p=2$ and for $p \neq 2$ achieves $O\left(\left(\frac{cp}{p-2}\right)^2\frac{ \log
      (1/\delta)}{N\eps^2}\right)$.
\end{theorem}

Notice that the bound on squared error improves on input and output
perturbation by a factor $\tilde{O}(\frac{1}{N})$.

The proof of Theorem~\ref{thm:compressible} follows from
Theorem~\ref{thm:CP-error} and the following lemma.

\begin{lemma}
  \label{lm:compr-bounds}
  Let $\bh$ be a $(c, p)$-compressible vector for some $p>1$. Then, we have
  \begin{equation*}
    \| \hat{\bh}\|_1 = \sum_{i = 0}^{N-1}{|\hat{h}_i|}  \leq
    \left\{
        \begin{array}{cc}
         c ( 1+  \ln N),  & \mbox{if } p=2\\
         \frac{c \: p}{p-2},  & \mbox{if } p>2
       \end{array}
       \right.
  \end{equation*}
  \junk{
  Then, for all $0 \leq a \leq N-1$, we have
  \begin{equation*}
    \sum_{i = a}^{N-1}{|\hat{h}_i|} \leq
    \left\{
        \begin{array}{cc}
         c + c \ln \frac {N}{a+1} & \mbox{if } p=2\\
         c + \frac{c}{(p/2-1) (a+1)^{p/2-1}} & \mbox{if } p>2
       \end{array}
       \right.
  \end{equation*}
  }
\end{lemma}
\begin{proof}
  Approximating a sum by an integral in the usual way, for $0 \leq a
  \leq b$ and $p\geq 2$, we have
  \begin{align*}
    \sum_{i = a}^b{\frac{1}{(i+1)^{p/2}}} &= \sum_{i =
      a+1}^{b+1}{\frac{1}{i^{p/2}}}\\ 
    &\leq \frac{1}{(a+1)^{p/2}} + \int_{a+1}^{b+1}{\frac{dx}{x^{p/2}}}
  \end{align*}
  Bounding the integral on the right hand side, we get
  \begin{equation*}
    \sum_{i = a}^b{\frac{1}{(i+1)^{p/2}}}
     \leq \left\{
        \begin{array}{cc}
         1 + \ln \frac{b+1}{a+1}, & \mbox{if }  p=2\\
         1 + \frac{1}{(p/2-1)(a+1)^{p/2-1}}, & \mbox{if }  p>2
       \end{array}
       \right.
  \end{equation*}
  The lemma then follows from the definition of $(c,  p)$-compressibility.
\end{proof}

\subsection{Running Sum}

Running sums can be defined as the circular convolution $x'\ \ast h$
of the sequences $h = (1, \ldots, 1, 0, \ldots, 0)$, where there are
$N$ ones and $N$ zeros, and $x' = (x, 0, \ldots, 0)$, where the private
input $x$ is padded with $N$ zeros. An elementary computation reveals
that $\hat{h}_1 = \sqrt{N}$ and $\hat{h}_i = O(N^{-1/2})$ for all
$i>1$. By Theorem~\ref{thm:CP-error}, Algorithm~\ref{alg:CP} computes
running sums with mean squared error $O(1)$ (ignoring dependence on
$\epsilon$ and $\delta$), improving on the bounds
of~\cite{chan2010private,dwork-continual,Xiao:2009db} in the mean
squared error regime.

\subsection{Linear Filters in Time Series Analysis}

Linear filtering is a fundamental tool in analysis of time-series
data. A time series is modeled as a sequence $x = (x_t)_{t =
  -\infty}^\infty$, supported on a finite set of time steps. A filter
converts the time series into another time series. A linear filter
does so by computing the convolution of $x$ with a series of
\emph{filter coefficients} $w$, i.e.~computing $y_t = \sum_{i =
  -\infty}^\infty{w_ix_{t-i}}$. For a finitely supported $x$, $y$ can
be computed using circular convolution by restricting $x$ to its
support set and padding with zeros on both sides.

We consider the case where $x$ is a time series of sensitive
\emph{events}. Each element $x_i$ is a count of events or sum of
values of individual transactions that have occurred at time step
$i$. When we deal with values of transactions, we assume that
individual transactions have much smaller value than the total.  We
emphasize that the definition of differential privacy with respect to
$x$ defined this way corresponds to \emph{event-level
privacy}. Semantically, this guarantee implies
that even an adversary who has arbitrary information about all 
but a single event of interest cannot find out with certainty whether
the event of interest has occur-ed. This guarantee is weaker than the
user-level guarantee, which implies that knowing all events related to
all  but a single user of interest provides little information
about the user. The
user-level guarantee would unfortunately require excessive noise for
filtering time series data, as the sensitivity of the convolution query
becomes unbounded. On the other hand, the event-level guarantee is
often sufficient, specifically in settings when sensitive events occur
only infrequently.

We consider applications to financial analysis, but our methods are
applicable to other instances of time series data, e.g.~we may also
consider network traffic logs or a time series of movie ratings on an
online movie streaming service.  We can perform almost optimal
differentially private linear filtering by casting the filter as a
circular convolution. Next we briefly describe a couple of
applications of private linear filtering to financial analysis. For
more references and detailed description, we refer the reader the book
of Gen\c{c}an, Sel\c{c}uk, and Whitcher~\cite{Gencay-2002}.

\textbf{Volatility Estimation}.  The value at risk measure is used to
estimate the potential change in the value of a good or financial
instrument. Assume, for
example, that in an online advertising system we would like to
estimate potential changes in the number of clicks per day for a set
of display ad campaigns, and denote by $x_i$ the number of clicks on
day $i$ from the start of the campaigns. The sensitive event is
assumed to be a single ad click, for example a click on an ad for a
type of medical treatment. In order to estimate volatility, we need to
estimate a measure of the deviation of the $x_i$ for a given time
period $[t-W+1, t]$. It is appropriate to take older fluctuations with
less significance. One way to do this is by using linear filtering of
the time series of absolute deviations in the click counts:
\begin{equation*}
  \sigma^e_t =  \frac{1}{\sum_{i = 1}^{W-1}{\lambda^i}}
    \sum_{i=0}^{W-1}{\lambda^i|x_{t-i} - \bar{x}_{t-i}|},
\end{equation*}
where $\lambda$ is a decay parameter and $\bar{x}_{t}$ is the average
count over $[t-W + 1, t]$. The quantity $\bar{x}_t$ is itself given by
the convolution $\frac{1}{W}\sum_{i = o}^{W-1}{x_{t-i}}$ and can be
computed nearly optimally using Algorithm~\ref{alg:CP}. Given the
sequence $\bar{x}$, we can construct the time series $(y_i)_i = (|x_i
- \bar{x}_i|)_{i}$. Using the triangle inequality, one can verify that
for a fixed value of $\bar{x}$, $\|y - y'\|_1 \leq \|x - x'\|_1$, and
therefore an algorithm which is differentially private with respect to
$y$ is also differentially private with respect to $x$. Therefore, we
can use Algorithm~\ref{alg:CP} to estimate $\sigma^e$ with nearly
optimal mean squared error.

Computing $\bar{x}$ was treated in~\cite{bolot2011private} as the
window sums problem, together with other decayed sum problems. The quantity $\sigma^e$ is
an exponentially decayed sum computed over a window and can be
approximated under $\eps$-differential privacy using the methods
of~\cite{bolot2011private}. However, as noted above,
Algorithm~\ref{alg:CP} gives improved mean squared error guarantees
for window sums, as well as a near-optimality guarantee\junk{: a direct computation of Fourier
coefficients and Theorem~\ref{thm:CP-error} imply that
Algorithm~\ref{alg:CP} gives $O(1)$ mean squared error for both
classes of queries}.

\textbf{Business Cycle Analysis}.  The goal of business cycle analysis
is to extract cyclic components in the time series and smooth-out
spurious fluctuation. Two classical methods for business-cycle
analysis are the Hodrick-Prescott filter and the Baxter-King
filter. Here we briefly sketch the form of the Hodrick-Prescott (HP)
filter. Let us take the example of time series $x$ of ad clicks again,
with a single component $x_i$ giving number of clicks on a set of ads
per day or per hour. We can use the HP filter to detect cyclical
trends in ad clicking activity. The filtered-out cyclical (smooth)
component of the data extracted by the HP filter can be written as a
convolution of the following form:
\begin{equation*}
  y_t^s = \frac{\theta_1\theta_2}{\lambda}\left(\sum_{j =
      0}^{\infty}{(A_1\theta_1^j + A_2\theta_2^j)(x_{t-j} + x_{t+j})}  \right).
\end{equation*}
Above, $\lambda$ is a smoothing parameter: the larger $\lambda$ is,
the more the data is smoothed by the filter; $\theta_i$ and $A_i$ are
functions of $\lambda$. In principle, this is a convolution of
infinite time series, but in practice we truncate the series to a
finite length. 


\subsection{Generalized Marginal Queries}\label{sect:gen-marg}

Marginal queries are a class of queries posed to $d$-attribute binary
databases, i.e.~databases where each row of the database is associted
with a $d$-bit binary vector, corresponding to the values of $d$
binary attributes. A marginal query is specified by a setting $\vec{a}
\in \{0,1\}^d$ of the $d$ attributes and a subset $S \subseteq [d]$ of
$k$ attributes; the exact answer to the query is the number of
rows in the database consistent with $\vec{a}$ on $S$. In this
subsection we address the error required to privately answer a natural
generalization of marginal queries. A generalized marginal query is
specified by a setting $\vec{a} \in \{0,1\}^d$ of the $d$ attributes
and a $w$-DNF $h$ and the exact answer is the number of rows $\vec{b}
\in \{0,1\}^d$ in the private database for which $h(\vec{a} \oplus
\vec{b})$ is satisfied (here $\oplus$ is componentwise XOR). In the
case of traditional marginal queries the DNF $h$ is a single
disjunction of $k$ unnegated variables. Generalized marginals however
allow more complex queries such as, for example, ``show all users who
agree with $\vec{a}$ on $a_1$ and at least one other attribute''.

More formally, we encode a binary $d$-attribute database in histogram
representation as a function $x:\{0,1\}^d \rightarrow [n]$. The
value of $x(\vec{a})$ for $\vec{a}\in \{0,1\}^d$ corresponds to the
number of rows in the database with attribute setting $\vec{a}$, and
$n$ is the database size. 

\begin{definition}
  Let $h(\vec{c})$ be a $w$-DNF given by $h(c) = (\ell_{1, 1}
  \wedge \ldots \wedge \ell_{1,w}) \vee \ldots \vee (\ell_{s, 1}
  \wedge \ldots \wedge \ell_{s, w})$, where $\ell_{i, j}$ is a
  literal, i.e.~either $c_p$ or $\bar{c_p}$ for some $p \in [d]$. The
  generalized marginal function for $h$ and a database $x:\{0, 1\}^d
  \rightarrow [n]$ is a function $(x \ast h):\{0, 1\}^d \rightarrow
  [n]$ defined by
  \begin{equation*}
    (x \ast h)(\vec{a}) = \sum_{b \in \{0, 1\}^d}{x(b)h(a \oplus b)}.
  \end{equation*}
\end{definition}

The overload of notation for $x \ast h$ here is on purpose as
generalized marginals can be interpreted as an instance of a
generalization of circular convolutions. In particular, circular
convolutions are associated naturally with the group of addition
modulo $N$, while generalized marginals are an instance of
convolutions associated with the group of addition modulo $2$ of
$d$-dimensional binary vectors (formally $(\Z/2\Z)^d$). Moreover,
there is a Fourier transform that diagonalizes convolutions over
$(\Z/2\Z)^d$ and that shares all properties with the transform defined
in Section~\ref{sect:prelims} which are necessary for our lower and
upper bound arguments. In particular, we need that any component of
any Fourier basis vector has norm $1/\sqrt{N}$, which is true for the
Fourier transform diagonalizing convolutions over
$(\Z/2\Z)^d$. Therefore, we can privately approximate generelized
marginal queries using Algorithm~\ref{alg:CP}, and, furthermore, our
analysis of the privacy and accuracy guarantees for the algorithm
still holds. Using results from learning theory on the spectral
concentration of bounded width DNFs and the bound from
Section~\ref{sect:compr}, we can show that Algorithm~\ref{alg:CP}
gives non-trivial error for generalized marginal queries.

\begin{theorem}\label{thm:gen-marg}
  Let $h$ be a $w$-DNF and $x:\{0, 1\}^d \rightarrow [n]$ be a private
  database. Algorithm~\ref{alg:CP} satisfies $(\eps,
  \delta)$-differential privacy and computes the generalized
  marginal $x \ast h$ for $h$ and and $x$ with mean squared error bounded by
  $O(\frac{\log(1/\delta)}{\eps^2}2^{d(1-1/O(w\log w))})$.
\end{theorem}

In addition to this explicit bound, we also know (by
Theorem~\ref{thm:alg-opt}) that up to a factor of $d^4$,
Algorithm~\ref{alg:CP} is optimal for computing generalized
marginal functions. Notice that error bound we proved improves on
randomized response by a factor of $2^{-\Omega(d/(w\log w))}$;
interestingly this factor is independent of the size of the $w$-DNF
formula.

In related work, Hardt et al.~\cite{hardt2012private} considered
database queries that can be computed by an AC0 circuit. Generalized
marginal queries can be computed by a two-layer AC0 circuit. However,
our results are incomparable to theirs, as they consider the setting
where the database is of bounded size $\|\vec{x}\|_1 \leq n$ and our
error bounds are independent of $\|x\|_1$. Our error
bounds improve on the bounds of~\cite{hardt2012private} when the
database is large enough so that our error bound is sublinear in
database size. 

The proof of Theorem~\ref{thm:gen-marg} follows from
Lemma~\ref{lm:compr-error} and the following concentration result for the
spectrum of $w$-DNF formulas, originally proved by
Mansour~\cite{mansour1995nlog} in the context of learning under the
uniform distribution.

\begin{theorem}[~\cite{mansour1995nlog}]
  Let $h:\{0, 1\}^d \rightarrow \{0, 1\}$ be a $w$-DNF. Let
  $\mathcal{F} \subseteq 2^{[d]}$ be the index set of the top
  $2^{d-k}$ Fourier coefficients of $h$. Then,
  \begin{equation*}
    \sum_{S \not \in \mathcal{F}}{|\hat{h}(S)|^2} \leq 2^{d +
      \frac{k-d}{O(w\log w)}}.
  \end{equation*}
\end{theorem}

\section{Conclusion} 

We derive nearly tight upper and lower bounds on the error of $(\eps,
\delta)$-differentially private for computing convolutions. Our lower
bounds rely on recent general lower bounds based on discrepancy theory
and elementary linear algebra; our upper bound is a simple
computationally efficient algorithm. We also sketch several
applications of private convolutions, in time series analysis and in
computing generalizes marginal queries on a $d$-attribute database.

Our results are nearly optimal for any $h$ when the database size is
large enough with respect to the number of queries. In some settings
it is reasonable to assume however that database size is much smaller,
and our algorithms give suboptimal error for such sparse
databases. Nearly optimal algorithms for  computing a workload of $M$ linear queries
posed to a  database of size at most $n$ were given
in~\cite{geometry}, but their algorithm has running time at least
$O(M^2Nn)$. Since our dense case algorithm for computing convolutions
has running time $O(N \log N)$, an interesting open problem is to give
an algorithm with running time $O(Nn \polylog(N,n))$ for computing
convolutions with optimal error when the database size is at most
$n$. 

\junk{
Note: Subsequently to our work,  quite recently, our results were generalized
by Nikolov, Talwar, and Zhang~\cite{geometry} to give instance-optimal
$(\eps, \delta)$-differentially a private algorithm for any workload
of linear queries. However, the generality of their result comes at a
cost in efficiency: their algorithm needs to compute a minimum
enclosing ellipsoid of a convex body and has running time at least
$O(M^2N)$. 
}

\bibliographystyle{acm}
\bibliography{./biblio/IEEEabrv,./biblio/DiffPrivacy}

\newpage
\onecolumn
\appendix
\section{Spectrum Partitioning Algorithm}
\label{app:specpart}

We partition the spectrum of the convolution matrix $\vec{H}$ into
geometrically growing in size groups and adds different amounts of
noise to each group. Noise is added in the Fourier domain, i.e.~to the
Fourier coefficients of the private input $\vec{x}$. The most noise is
added to those Fourier coefficients which correspond to small (in
absolute value) coefficients of $\vec{h}$, making sure that privacy is
satisfied while the least amount of noise is added. In the analysis of
optimality, we show that the noise added to each group can be charged
to the lower bound $\specLB(\vec{h})$. Because the number of groups is
logarithmic in $N$, we get almost optimality. This analysis is
inspired by the work of Hardt and
Talwar\cite{hardt2010geometry}. However, our algorithm is simpler and
significantly more efficient.

The $(\eps, \delta)$-differentially private algorithm we propose for
approximating $h \ast x$ is shown as
Algorithm~\ref{alg:compressible}. In the remainder of this section we
assume for simplicity that $N$ is a power of 2. We also assume, for
ease of notation, that $|\hat{h}_0| \geq \ldots \geq |\hat{h}_{N-1}|$. Our
algorithm and analysis do not depend on $i$ except as an index, so
this comes without loss of generality.

\begin{algorithm}[h]
  \caption{\textsc{SpectralPartition}}\label{alg:compressible}
  \begin{algorithmic}
    \STATE Set $\eta = \frac{\sqrt{2(1+\log N)\ln(1/\delta)}}{\eps}$
    \STATE Compute $\vec{\hat{x}} = \bF_N\vec{x}$ and $\vec{\hat{h}} =
    \bF_N \vec{x}$.
    \STATE $\tilde{x}_0 = \hat{x}_0 + \Lap(\eta)$
    \FORALL{$k \in [1, \log N]$ }
    \FORALL{$i \in [N/2^k, N/2^{k-1} - 1]$}
    \STATE Set $\tilde{x}_i = \hat{x}_i + \Lap(\eta 2^{-k/2})$.
    \STATE Set $\bar{y}_i = \sqrt{N}\hat{h}_i \tilde{x}_i$.
    \ENDFOR
    \ENDFOR
    \STATE Output $\vec{\tilde{y}} =  \bF_N^H \vec{\bar{y}}$
  \end{algorithmic}
\end{algorithm}

\begin{lemma}
  \label{lm:compr-error}
  Algorithm~\ref{alg:compressible} satisfies $(\eps,
  \delta)$-differential privacy. Also, there exists an absolute
  constant $C$ such that Algorithm~\ref{alg:compressible} achieves
  expected mean squared error
  \begin{equation}\label{eq:accuracy}
    \MSE \leq C\frac{(1+\log N) \log
      (1/\delta)}{\eps^2}(|\hat{h}_0|^2 +  \sum_{k = 1}^{\log
      N}{\frac{1}{2^k}\sum_{i = N/2^k}^{N/2^{k-1} - 1}{|\hat{h}_i|^2}}).
  \end{equation}
\end{lemma}
\begin{proof}
  \textbf{Privacy.} We claim that $\vec{\tilde{x}}$ is an $(\eps,
  \delta)$-differentially private function of $\vec{x}$. The other
  computations depend only on $\vec{h}$ and $\vec{\tilde{x}}$ and not
  on $\vec{x}$ directly, so, by Theorem~\ref{thm:simple-composition},
  incur no loss in privacy.

  First we analyze the sensitivity of each Fourier coefficient
  $\hat{x}_i$. As a function of $\vec{x}$, $\hat{x}_i$ is an inner
  product of $\vec{x}$ with a Fourier basis vector. Let that vector be
  $\vec{f}$ and let $\vec{x}$, $\vec{x}'$ be two neighboring inputs,
  i.e.~$\|\vec{x - x'}\|_1 \leq 1$. Then we have
  \begin{equation*}
    |\vec{f}^H(\vec{x} - \vec{x'})| \leq \|\vec{f}\|_\infty \|\vec{x -
      x'}\|_1 \leq \frac{1}{\sqrt{N}}
  \end{equation*}
  Therefore, by Theorem~\ref{th:eps0-DiffPriv}, when $i \in [N/2^k,
  N/2^{k-1}-1]$, $\tilde{x}_i$
  is $(\frac{2^{k/2}}{\sqrt{N}\eta}, 0)$-differentially private. By
  Theorem~\ref{thm:fancy-comp}, $\vec{\tilde{x}}$ is $(\eps',
  \delta)$ differentially private for any $\delta > 0$, where
  \begin{align*}
    \eps'^2 &= 2\ln(1/\delta)(\frac{1}{\eta^2} + \sum_{k = 1}^{\log
      N}{\frac{N}{2^k}\frac{2^k}{N\eta^2}})\\
    &= 2\ln(1/\delta)\frac{1+\log N}{\eta^2} = \eps^2
  \end{align*}

  \textbf{Accuracy.} Observe $\E[\tilde{x}_i] = \hat{x}_i$ since we
  add unbiased Laplace noise to each $\hat{x}_i$. Also, the variance
  of $\Lap(\eta 2^{-k/2})$ is $2\eta^2 2^{-k}$. Therefore,
  $\E[\bar{y}_i] = \sqrt{N}\hat{h}_i\hat{x_i}$ and the variance of
  $\bar{y}_i$ when $i \in [N/2^k, N/2^{k-1}-1]$ is $O(N|\hat{h}_i|^2\eta^2
  2^{-k})$. By linearity of expectation, $\E[\bF_N^H\vec{\bar{y}}] =
  \bH\bx$. Adding variances for each $k$ and dividing by $N$, we get
  the right hand side of (\ref{eq:accuracy}). The proof is completed
  by observing that the inverse Fourier transform $\bF^H_N$ is an
  isometry for the $\ell_2$ norm, so does not change mean squared
  error.
\end{proof}

\begin{theorem}\label{thm:alg-opt}
  For any $\vec{h}$, Algorithm ~\ref{alg:compressible} satisfies $(\eps,
  \delta)$-differential privacy and achieves expected mean squared
  error $O(\specLB(\vec{h})\frac{\log^4N \ln (1/\delta)}{\eps^2})$.
\end{theorem}
\begin{proof}
  By Lemma~\ref{lm:compr-error}, we know that
  \begin{align*}
    \MSE \junk{&\leq C\frac{\log N \log
      (1/\delta)}{\eps^2}(|\hat{h}_0|^2 + \sum_{k = 1}^{\log
      N}{\frac{1}{2^k}\sum_{i = N/2^k}^{N/2^{k-1}-1}{|\hat{h}_i|^2}})\\}
    \leq C\frac{\log N \log (1/\delta)}{\eps^2}(|\hat{h}_0|^2 +  \sum_{k = 1}^{\log
      N}{\frac{N}{2^{2k}}|\hat{h}_{N/2^{k-1}-1}|^2})
    \junk{& = \frac{O(\log^3 N) \log (1/\delta)}{\eps^2} \sum_{k = 1}^{\log
      N}{\specLB(\vec{h})}\\}
    = O(\specLB(\vec{h})\frac{\log^4N \ln (1/\delta)}{\eps^2}).
  \end{align*}
\end{proof}

\junk{Theorem~\ref{thm:alg-opt} shows that Algorithm~\ref{alg:compressible}
is almost optimal \emph{for any given $\vec{h}$}. We also compute
explicit asymptotic error bounds for a particular case of interest,
compressible $\vec{h}$ as per definition~\ref{def:compressible}, for which Algorithm~\ref{alg:compressible}
outperforms input and output perturbation.


\begin{lemma}
  \label{lm:compr-bounds-2}
  Let $h$ be a $(c, p)$-compressible vector for some $p>1$. Then, for all $0\leq a \leq N-1 $, we have
  \begin{equation*}
    \sum_{i = a}^{N-1}{|\hat{h}_i|^2} \leq c + \frac{c}{(p-1)(a+1)^{p-1}}.
  \end{equation*}
\end{lemma}
\begin{proof}
  Approximating a sum by an integral in the usual way, for $0 \leq a
  \leq b$ and $p>1$, we have
  \begin{equation*}
    \sum_{i = a}^b{\frac{1}{(i+1)^p}}= \sum_{i = a+1}^{b+1}{\frac{1}{i^p}}\leq \frac{1}{(a+1)^p} + \int_{a+1}^{b+1}{\frac{dx}{x^p}}
    = \frac{1}{(a+1)^p} + \frac{1}{p-1}\left(\frac{1}{(a+1)^{p-1}} -
      \frac{1}{b^{p-1}}\right) \leq 1 + \frac{1}{(p-1)(a+1)^{p-1}}.
  \end{equation*}
  The lemma then follows form the definition of $(c,  p)$-compressibility.
\end{proof}

The following two theorems follow from Lemma~\ref{lm:compr-error} and
Lemma~\ref{lm:compr-bounds}.
\begin{theorem}
  Let $\vec{h}$ be a $(c, 2)$-compressible vector. Then
  Algorithm~\ref{alg:compressible} satisfies $(\eps,
  \delta)$-differential privacy and achieves expected mean squared error
  $O(\frac{c\log^2 N \log (1/\delta)}{N\eps^2})$.
\end{theorem}

\begin{theorem}
  Let $\vec{h}$ be a $(c, p)$-compressible vector for some constant $p
  > 2$. Then Algorithm~\ref{alg:compressible} satisfies $(\eps,
  \delta)$-differential privacy and achieves expected mean squared error
  $O(\frac{c\log N \log (1/\delta)}{N\eps^2})$.
\end{theorem}
}
\junk{By Theorem~\ref{thm:main-lb}, when $|\hat{h}_1| = c$ and $\eps$,
$\delta$ are sufficiently small constants, any $(\eps,
\delta)$-differentially private algorithm incurs with mean squared
error $\sigma_z = \Omega(\frac{c}{N})$. Therefore,
Algorithm~\ref{alg:compressible} is almost optimal for $(c,
p)$-compressible $\vec{h}$ when $p\geq 2$. }

\section{Closed Form Solution for the Optimal $\mathcal{A}(\vec{b})$}
\label{app:KKT}

We derive a closed form solution of (\ref{eq:obj})--(\ref{eq:non-neg})
using convex programming duality. Let us first rewrite the program by substituting
$a_i = 1/b_i^2$:
\begin{equation}\label{pb:CP}
\begin{split}
\min_{\{a_i\}_{i\in I}} & \sum_{i \in I} \frac{|\hat{h}_i|^2}{a_i} \\
\mbox{s.t.} &  \sum_{i \in I} a_i = \frac{N \eps^2}{2 \ln (1/\delta)}\\
& a_i\geq 0, \quad \forall i \in I. 
\end{split}
\end{equation}
The Lagrangian is
\begin{equation}
L(\ba,\nu,\Lambda)= \sum_{i \in I} \frac{|\hat{h}_i|^2}{a_i} + \nu \left( \sum_{i \in I} a_i - \frac{N \eps^2}{2 \ln (1/\delta)} \right) - \sum_{i \in I} \lambda_i a_i.
\end{equation}
The KKT conditions are given by 
\begin{equation}
\begin{split}
\forall i \in I, \quad & - \frac{|\hat{h}_i|^2}{a_i^2} + \nu - \lambda_i = 0\\
& \sum_{i \in I} a_i - \frac{N \eps^2}{2 \ln (1/\delta)}=0\\
& \lambda_i a_i =0\\
& a_i \geq 0, \lambda_i \geq 0
\end{split}
\end{equation}
The following solution $(\ba^*,\nu^*,\Lambda^*)$ satisfies the KKT conditions, and is thus the optimal solution to (\ref{pb:CP})
\begin{equation}
\begin{split}
\forall i \in I, \qquad & a^*_i=  \frac{N \eps^2}{2 \ln (1/\delta) \| \hat{\bh}\|_1}  |\hat{h}_i|,
\qquad \qquad \lambda^*_i=0,
\qquad \qquad \nu^*= \left(\frac{2 \ln (1/\delta) \|\hat{\bh}\|_1}{N \eps^2}\right)^2.
\end{split}
\end{equation}
Consequently, the optimal noise parameters $\bb$ for the original problem (\ref{eq:obj})--(\ref{eq:non-neg}), and the associated MSE are
\begin{equation}
\begin{split}
b_i^*&= \left\{
\begin{array}{cc}
  \sqrt{\frac{2 \ln (1/\delta) \| \hat{\bh}\|_1}{N \eps^2 |\hat{h}_i| }}  & \mbox{if } i \in I\\
  0 & \mbox{if } i \notin I
\end{array}
\right.\\
\MSE^*&=  2 \sum_{i\in I} |\hat{h}_i|^2 b_i^2 = 4  \frac{\ln(1/\delta)}{\eps^2 N}\| \vec{\hat{h}} \|_1^2,
\end{split}
\end{equation}
which are the noise parameters and $\MSE$ of Algorithm~\ref{alg:CP}. 

\end{document}